\documentclass[journal,onecolumn,12pt]{IEEEtran}
\usepackage{amsmath}
\usepackage{amssymb}
\usepackage{amsthm}
\usepackage{amsmath}
\usepackage{amssymb}
\usepackage{url}
\usepackage{amscd}
\usepackage{mathrsfs}
\usepackage{graphicx}
\usepackage{float}

\usepackage{amsmath}

\newcommand{\isep}{\!\mathrel{{.}{.}\!}\nobreak}

\usepackage{setspace}
\usepackage[OT2,T1]{fontenc}
\DeclareSymbolFont{cyrletters}{OT2}{wncyr}{m}{n}
\DeclareMathSymbol{\Sha}{\mathalpha}{cyrletters}{"58}

\newtheorem{theorem}{Theorem}

\newtheorem{lemma}{Lemma}

\newtheorem{definition}{Definition}

\newcommand{\Fq}{\mathbb{F}_q}

\newcommand{\R}{\mathbb{R}}
\newcommand{\E}{\mathbb{E}}

\newcommand{\A}{\mathbf{A}}
\newcommand{\tE}{\tilde{\E}}
\newcommand{\bC}{\mathbb{C}}
\newcommand{\db}{d^\bot}
\newcommand{\Op}{\Omega_p}
\newcommand{\OpI}{\Omega_{p,I}}
\newcommand{\om}{\omega}
\newcommand{\Om}{\Omega}
\newcommand{\C}{\mathcal{C}}

\newcommand{\D}{\mathcal{D}}
\newcommand{\G}{\mathcal{G}}

\newcommand{\wt}{\mathrm{wt}}
\newcommand{\tr}{\mathrm{Tr}}

\newcommand{\wG}{\widetilde{\mathcal{G}}}
\newcommand{\hG}{\mathcal{G}_{\C_i,I}}
\newcommand{\wM}{\widetilde{M}_{\C_i}}

\newcommand{\SC}{M_\mathrm{SC}}
\newcommand{\MP}{M_{\mathrm{MP},y}}

\newcommand{\bg}{\mathbf{g}}

\newcommand{\bz}{\mathbf{0}}
\newcommand{\bo}{\mathbf{1}}
\renewcommand{\l}{\ell}

\newcommand{\ga}{\gamma}
\newcommand{\Ga}{\Gamma}
\newcommand{\ep}{\epsilon}
\newcommand{\Vg}{V_\ga}
\newcommand{\vg}{v_\ga}
\newcommand{\Vgu}{V_{\ga_1,\ga_2}}
\newcommand{\vgu}{v_{\ga_1,\ga_2}}
\newcommand{\Vgi}{V_{\ga_1\cap\ga_2}}
\newcommand{\vgi}{v_{\ga_1\cap\ga_2}}
\newcommand{\vgii}{v_{\ga_1,\ga_2}}
\newcommand{\Wg}{W_\ga}
\newcommand{\Wgu}{W_{\ga_1,\ga_2}}
\newcommand{\Wguu}{W^{\ga_1,\ga_2}}
\newcommand{\Og}{\Omega(\Vg)}
\newcommand{\OIg}{\Omega_I(\Vg)}
\newcommand{\OIgg}{\Omega_I(V_{\ga_1,\ga_2})}
\newcommand{\Ogg}{\Omega (V_{\ga_1,\ga_2})}

\newcommand{\AIl}{A_{\l,I}}

\newcommand{\la}{\lambda}
\newcommand{\Sp}{\Sigma_p}

\newcommand{\tb}{\textbf}
\newcommand{\all}{\forall}
\newcommand{\de}{\delta}
\newcommand{\nono}{\nonumber}
\newcommand{\gT}{\ga_T}

\newcommand{\setm}{\setminus}

\newcommand{\Var}{\mathrm{Var}}

\title{Random Matrices from Linear Codes and Wigner's semicircle law}
\author{Chin Hei Chan\thanks{C. Chan is at the Dept. of Mathematics, Hong Kong University of Science and Technology, Clear Water Bay, Kowloon, Hong Kong (email: chchanam@connect.ust.hk).}, Enoch Kung\thanks{E. Kung is at the Clinical Operational Research Unit, Dept. of Mathematics,
Faculty of Maths \& Physical Sciences, University College London, UK (email: e.kung@ucl.ac.uk).} and Maosheng Xiong\thanks{M. Xiong is at the Dept. of Mathematics, Hong Kong University of Science and Technology, Clear Water Bay, Kowloon, Hong Kong (email: mamsxiong@ust.hk). }}
\begin{document}
\maketitle
\begin{abstract}
In this paper we consider a new normalization of matrices obtained by choosing distinct codewords at random from linear codes over finite fields and find that under some natural algebraic conditions of the codes their empirical spectral distribution converges to Wigner's semicircle law as the length of the codes goes to infinity. One such condition is that the dual distance of the codes is at least 5. This is analogous to previous work on the empirical spectral distribution of similar matrices obtained in this fashion that converges to the Marchenko-Pastur law.
\end{abstract}

\section{Introduction}


The theory of random matrices mainly concerns the statistical behavior of eigenvalues of large random matrices arising from various matrix models. There is a universality phenomenon that, like the law of large numbers in probability theory, the collective behavior of eigenvalues of a large random matrix does not depend on the distribution details of entries of the matrix. Partly because of this reason, originated from statistics \cite{Wis} and mathematical physics \cite{Wig} and nurtured by mathematicians, the random matrix theory has found important applications in many diverse disciplines such as number theory \cite{MET}, computer science, economics and communication theory \cite{TUL} and remains a prominent research area.

Most of the matrix models considered in the literature were matrices whose entries have independent structures. In a series of work (\cite{Babadi1,Babadi2,Xia}), initiated in \cite{Babadi0}, the authors studied matrices formed from linear codes over finite fields and ultimately proved that they behave like truly random matrices (i.e., random matrices with i.i.d. entries) in terms of the empirical spectral distribution, if the minimum Hamming distance of the dual codes is at least 5. This is the first result relating the randomness of matrices from linear codes to the algebraic properties of the underlying dual codes, and can be interpreted as a joint randomness test for codes or sequences. This is called a ``group randomness'' property \cite{Babadi0} and may have many applications.

In this paper we study a new group randomness property of linear codes. To describe our results, we need some notation.


Let $\mathscr{C}=\{\C_i : i \ge 1\}$ be a family of linear codes of length $n_i$, dimension $k_i$ and minimum Hamming distance $d_i$ over the finite field $\Fq$ of $q$ elements ($\C_i$ is called an $[n_i,k_i,d_i]_q$ code for short). Assume that $n_i \to \infty$ as $i \to \infty$. The standard additive character on the finite field $\Fq$ extends component-wise to a natural mapping $\ep: \Fq^{n} \to \bC^{n}$. For each $i$, choosing $p_i$ codewords at random uniformly from $\C_i$ and applying the mapping $\ep$, we obtain a $p_i \times n_i$ random matrix $\Phi_{\C_i}$. The Gram matrix of $\frac{1}{\sqrt{n_i}}\Phi_{\C_i}$ is
\[\G_{\C_i}:=\frac{1}{n_i} \Phi_{\C_i} \Phi_{\C_i}^*,\]
here $\Phi_{\C_i}^*$ denotes the conjugate transpose of $\Phi_{\C_i}$. Denote by $\E$ the expectation with respect to the probability space.

For any $n \times n$ matrix $\A$ with eigenvalues $\lambda_1,\ldots,\lambda_n$, the \emph{spectral measure} of $\A$ is defined by
$$\mu_\A=\frac{1}{n}\sum_{j=1}^n \de_{\la_j},$$
where $\de_{\lambda}$ is the Dirac measure at the point $\lambda$. The \emph{empirical spectral distribution} of $\A$ is defined as
$$M_\A(x):=\int_{-\infty}^x \mu_\A(\text{d}x).$$
For the sake of brevity, a slightly simplified version of \cite[Theorem 1]{Xia} may be stated as follows.

\begin{theorem} \label{thm:xia}
Let $M_{\C_i}(x)$ be the empirical spectral distribution of the Gram matrix $\G_{\C_i}$. If the dual distance of the code $\C_i$ satisfies $\db_i \geq 5$ for each $i$ and $y=\frac{p_i}{n_i} \in (0,1)$ is fixed, then for any $x \in \R$, we have
\begin{equation} \label{1:xia}
\lim_{n_i \to \infty} \E M_{\C_i}(x) = \MP(x).
\end{equation}
\end{theorem}
\noindent Here $\MP(x)$ denotes the cumulative distribution function of the Marchenko-Pastur measure whose density function is given by
\begin{equation*} \label{eq:MP}
\rho_{_{\mathrm{MP},y}}(x):=\frac{1}{2\pi xy}\sqrt{(b-x)(x-a)}\bo_{[a,b]}(x),
\end{equation*}
where $a=(1-\sqrt{y})^2, b=(1+\sqrt{y})^2$, and $\bo_{[a,b]}$ is the indicator function of the interval $[a,b]$.

It is well-known in random matrix theory that, if $X_n$ is a $p \times n$ matrix whose entries are i.i.d. random variables of zero mean and unit variance, the empirical spectral distribution of the Gram matrix of $\frac{1}{\sqrt{n}}X_n$ satisfies the same Marchenko-Pastur law (\ref{1:xia}) as $n \to \infty$ and $y=\frac{p}{n}$ is fixed (see \cite{RM,MP}), hence the above result can be interpreted as that matrices formed from linear codes of dual distance at least 5 behave like truly random matrices of i.i.d. entries. In other words, sequences from linear codes of dual distance at least 5 possess a group randomness property. The condition $\db_i \geq 5$ is also necessary, because the empirical spectral distribution of matrices formed from the first-order Reed-Muller codes whose dual distance is 4 behave very differently from the Marchenko-Pastur law (\cite{Babadi0}).

In this paper we consider a different group randomness property. If $X_n$ is a $p \times n$ random matrix whose entries are i.i.d. random variables of zero mean and unit variance, let $G_n:=\frac{1}{n}X_nX_n^*$, it is well-known in random matrix theory (\cite{RM,Bao}) that in the limit $n,p,\frac{n}{p} \to \infty$ simultaneously, the empirical spectral distribution of the matrix $G_{n,I}:=\sqrt{\frac{n}{p}}(G_n-I_p)$ converges to Wigner's semicircle law $\SC(x)$ whose density function is given by
\begin{equation*} \label{eq:SC}
\rho{_{_\mathrm{SC}}}(x):=\frac{1}{2\pi}\sqrt{4-x^2}\cdot \bo_{[-2,2]}(x).
\end{equation*} Here $I_{p}$ denotes the identity matrix of size $p$. So a natural question is to investigate when similarly formed matrices from linear codes $\C_i$ satisfy the same property. For this purpose, we consider the $p_i \times n_i$ random matrix $\widetilde{\Phi}_{\C_i}$ obtained by choosing $p_i$ distinct codewords at random uniformly from $\C_i$ and by applying the mapping $\ep$. Define
\begin{equation*} \label{1:gc2} \G_{\C_i,I}:=\sqrt{\frac{n_i}{p_i}}(\wG_{\C_i}-I_{p_i}).
\end{equation*}

Now we state the main result of this paper.
\begin{theorem} \label{thm:SC}
Let $\wM(x)$ be the empirical spectral distribution of the matrix $\hG$. Assume that the linear codes $\C_i$ satisfy:

(i) $\frac{N_i}{n_i} \to \infty$ as $i \to \infty$, where $N_i=q^{k_i}$ is the cardinality of the code $\C_i$;

(ii) $\db_i \geq 5$ for each $i$, and

(iii) there is a fixed constant $c > 0$ independent of $i$ such that
\begin{equation} \label{1:srip}
|\langle v,v' \rangle| \leq c\sqrt{n_i}, \quad \mbox{ for any } v \ne v' \in \ep(\C_i).
\end{equation}
Here $\langle v,v' \rangle$ is the standard inner product of the complex vectors $v$ and $v'$. Then as $n_i,p_i, \frac{n_i}{p_i} \to \infty$ simultaneously, for any $x \in \R$, we have
\begin{eqnarray*} \label{1:conv-P} \wM(x) &\to & \SC(x) \quad \emph{ in Probability}. \end{eqnarray*}
\end{theorem}


We remark that condition (iii) is quite natural for linear codes, for instance, it appeared as a requirement in the construction of deterministic sensing matrices from linear codes that satisfy the ideal Statistical Restricted Isometry Property (see \cite[Definition 1]{DSM} or \cite{DSM2}). For binary linear codes $\C$ of length $n$, (iii) is equivalent to the condition
\[\left|\wt(\underline{c})-\frac{n}{2} \right| \le \frac{c}{2} \sqrt{n} \]
for any nonzero codeword $\underline{c} \in \C$. Here $\wt(\underline{c})$ is the Hamming weight of the codeword $\underline{c}$. There is an abundance of binary linear codes that satisfy this condition, for example, the Gold codes (\cite{GOL}), some families of BCH codes (see \cite{DSM,DIN1,DIN2}, and many families of cyclic and linear codes studied in the literature (see for example \cite{CHE,TAN,Xiong2}).

Next, we emphasis that in Theorem \ref{thm:SC} we prove the convergence ``in probability''. This is not only stronger than say $\E  \wM(x) \to  \SC(x)$ in probability theory (compared with Theorem \ref{thm:xia}) (see \cite{PTA}), but also much more useful in practice: it implies that under the conditions (i)-(iii), if $n_i$ is relatively large, then for any fixed $x$,  randomly choosing $p_i$ codewords from $\C_i$, then for most of the case, the resulting function $\wM(x)$ will be very close to the value $\SC(x)$. This can be easily confirmed by numerical experiments. We focus on binary Gold codes which have length $n=2^m-1$ and dual distance 5. Binary Gold codes satisfy the condition (\ref{1:srip}) because there are only three nonzero weights, namely $2^{m-1}-2^{(m-1)/2}, 2^{m-1}$ and $2^{m-1}+2^{(m-1)/2}$. Also the Gold codes have dimension $2m$ and so $\frac{n}{N}=\frac{2^m}{2^{2m}} \to 0$ as $m \to \infty$. For each pair $(n,p)$ in the set $\left\{(31,8),(127,20),(511,35), (2047,50)\right\}$, we randomly pick $p$ codewords from the binary Gold code of length $n$ and form the corresponding matrix, from which we compute and plot the empirical spectral distribution together with Wigner's distribution (see Figures \ref{fig1} to \ref{fig4} below). We do it 10 times for each such pair $(n,p)$ and at each time, we find that the plots are almost the same as before: they are all very close to Wigner's semicircle law and as the length $n$ increases, they become more and more indistinguishable.

\begin{figure}[ht]
\begin{center}
\includegraphics[angle=0,width=1.0 \textwidth,height=0.2 \textheight]{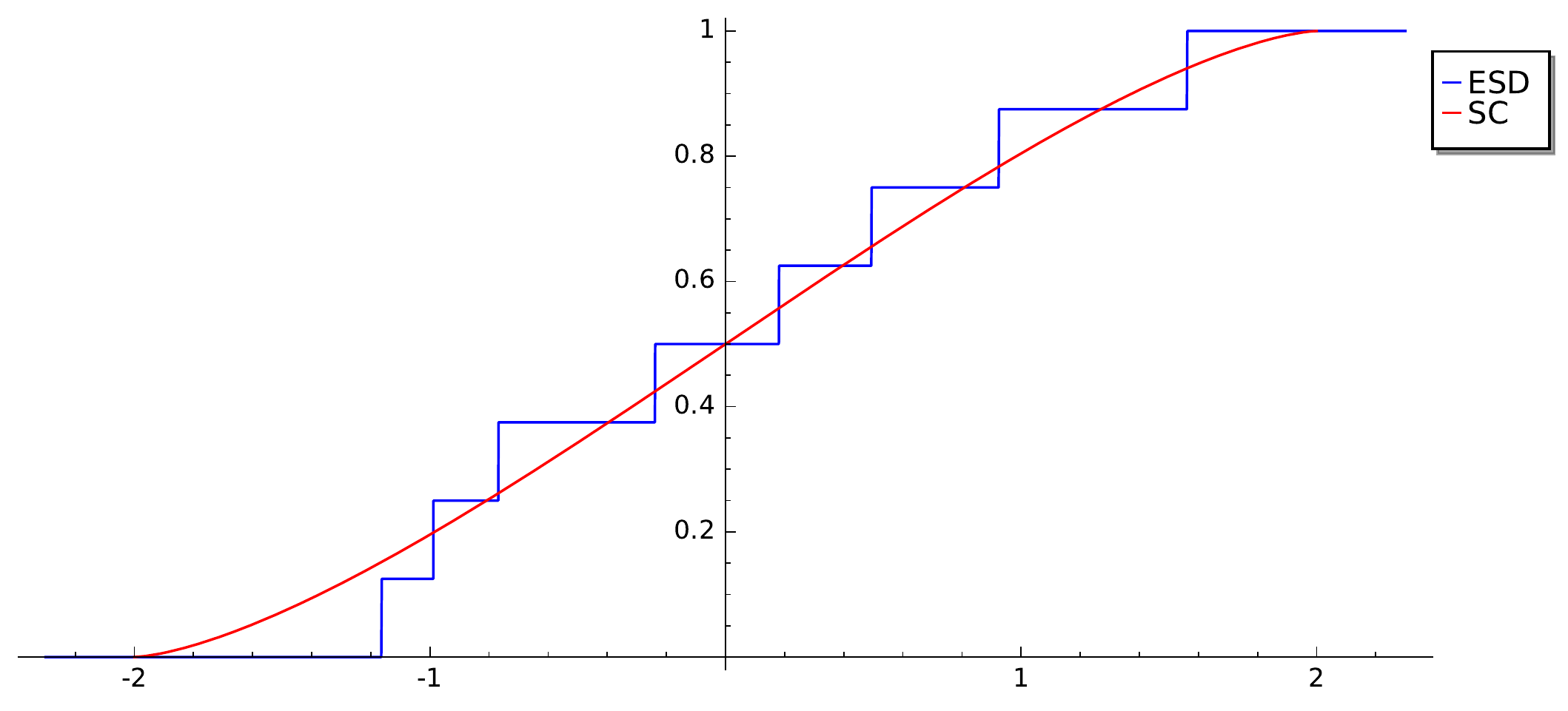}
\caption{Empirical spectral distribution (ESD) of $[31,10,12]$ binary Gold code versus Wigner semicircle law (SC), with $p=8, \db=5$}\label{fig1}
\end{center}
\end{figure}
\begin{figure}[ht]
\begin{center}
\includegraphics[angle=0,width=1.0 \textwidth,height=0.2 \textheight]{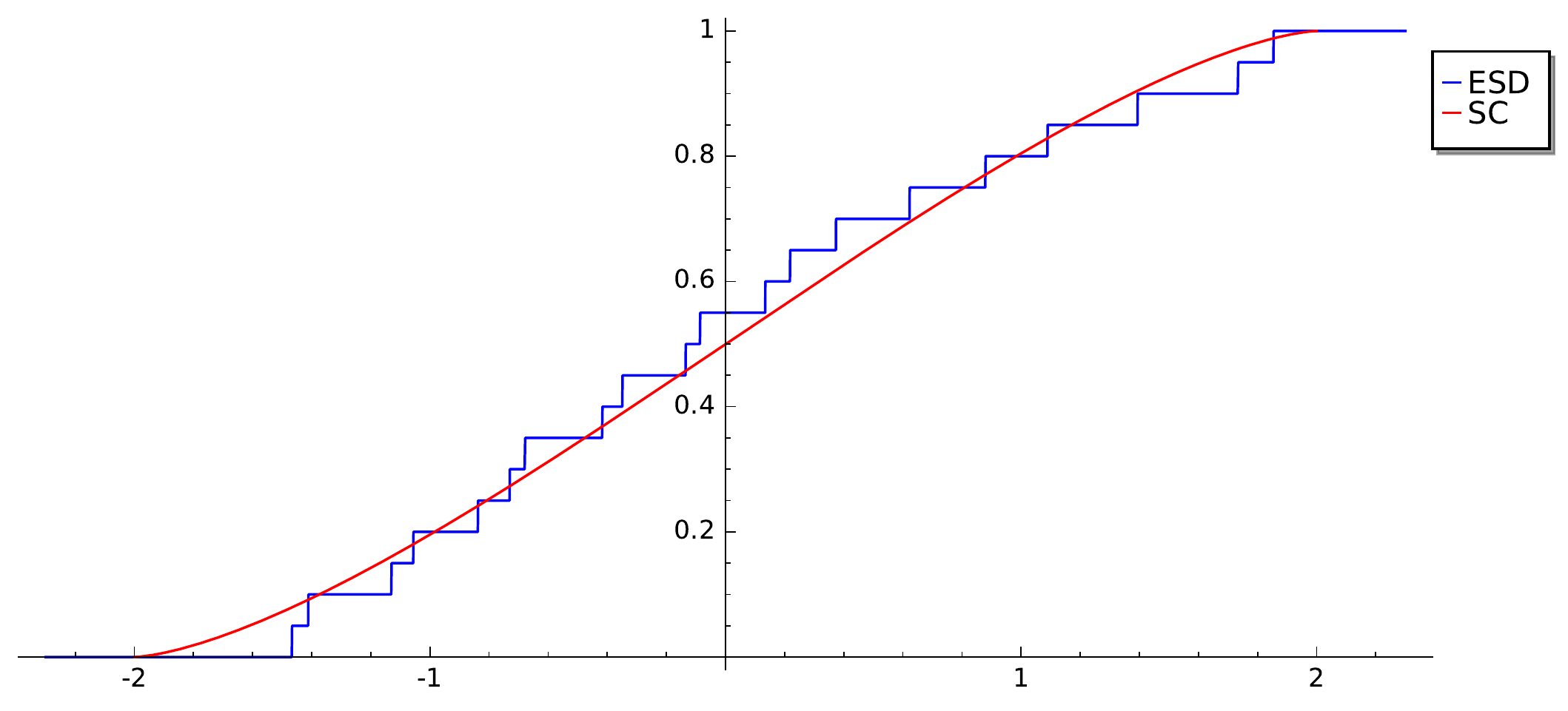}
\caption{Empirical spectral distribution (ESD) of $[127,14,56]$ binary Gold code versus Wigner semicircle law (SC), with $p=20, \db=5$}\label{fig2}
\end{center}
\end{figure}
\begin{figure}[ht]
\begin{center}
\includegraphics[angle=0,width=1.0 \textwidth,height=0.2 \textheight]{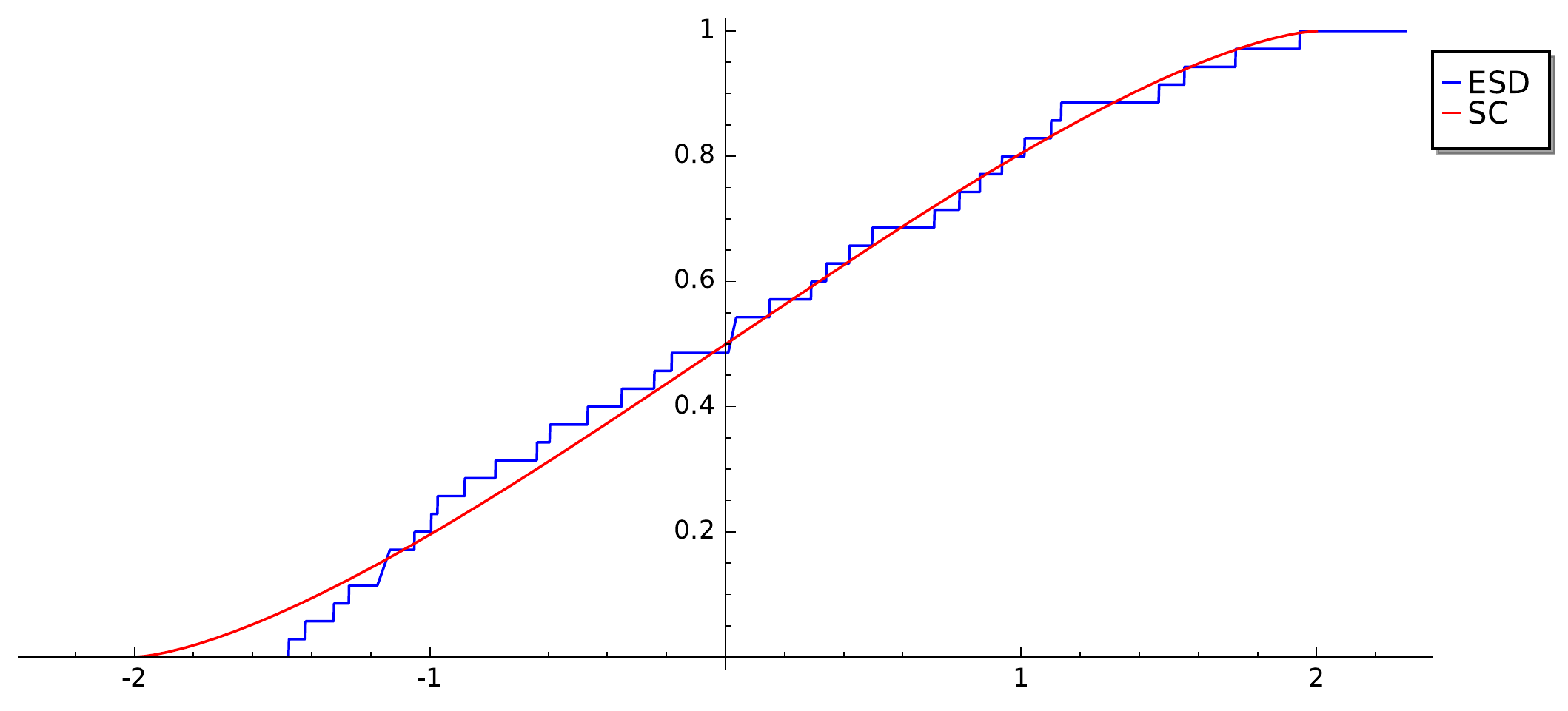}
\caption{Empirical spectral distribution (ESD) of $[511,18,240]$ binary Gold code versus Wigner semicircle law (SC), with $p=35, \db=5$}\label{fig3}
\end{center}
\end{figure}
\begin{figure}[ht]
\begin{center}
\includegraphics[angle=0,width=1.0 \textwidth,height=0.2 \textheight]{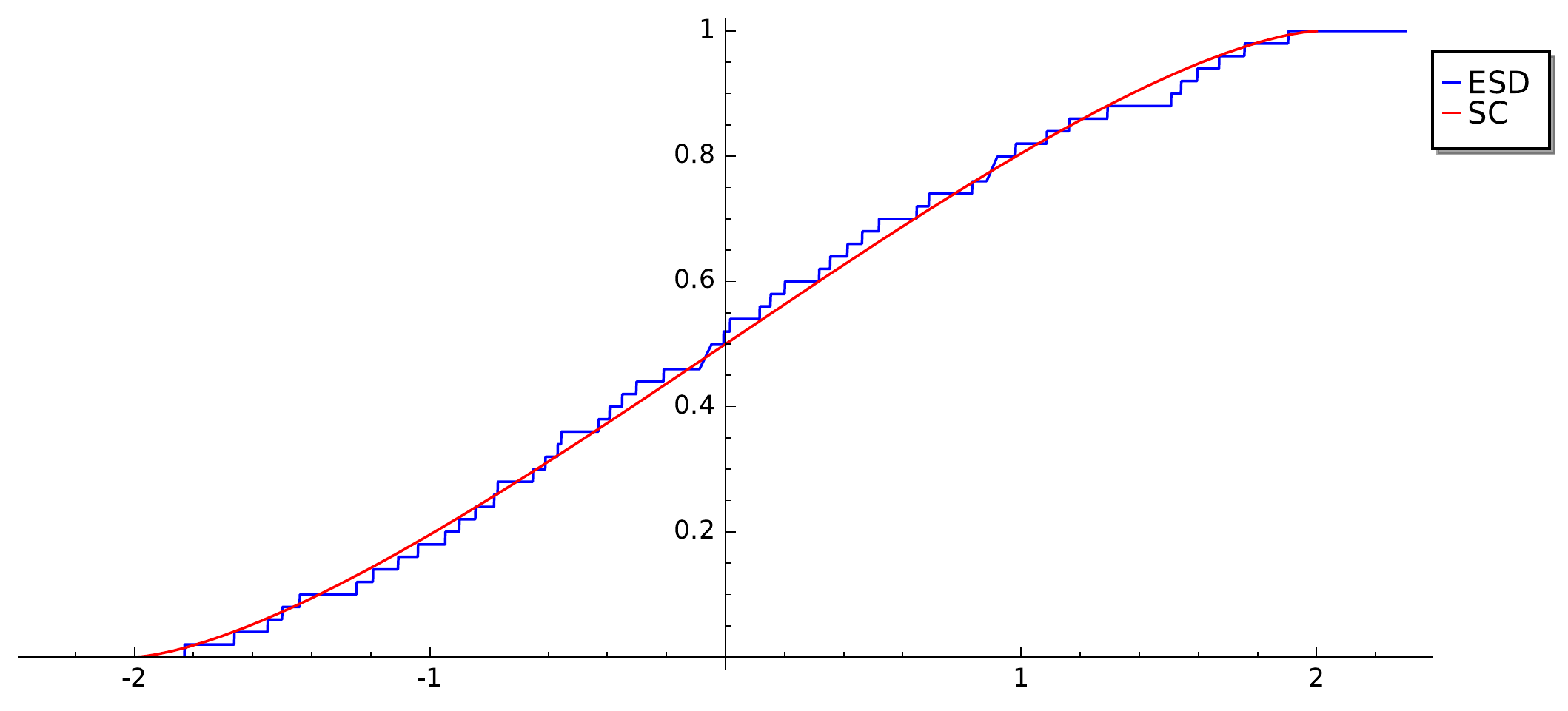}
\caption{Empirical spectral distribution (ESD) of $[2047,22,992]$ binary Gold code versus Wigner semicircle law (SC), with $p=50, \db=5$}\label{fig4}
\end{center}
\end{figure}

To prove Theorem \ref{thm:SC}, we use the moment method, that is, we compute the moments and the variance for the empirical spectral distribution and compare them with Wigner's semicircle law. This is a standard method in random matrix theory and has been used in \cite{Babadi2,Xia}. We mainly follow the ideas and techniques from \cite{Xia}. However, compared with \cite{Xia}, due to the nature of the problem, the computation, especially the variance becomes much more complicated. In order to present the ideas of the proof of Theorem \ref{thm:SC} more clearly, in Section \ref{sec:2} we sketch the main steps of the proof of Theorem \ref{thm:xia} in \cite{Xia}. This will serve as a general guideline for the proofs later on; We also prove some counting lemmas which will be used later. In Section \ref{sec:3} we compute the required moments with respect to Wigner's semicircle law, and in Section \ref{sec:4} we study the variance. This concludes the proof of Theorem \ref{thm:SC}. Sections \ref{sec:3} and \ref{sec:4} require the use of some crucial but technical lemmas. In order to present the ideas of the proofs more transparently, we postpone the proofs of those lemmas in Section \ref{appendix} {\bf Appendix}. Finally in Section \ref{sec:5} we conclude the paper.

\section{Preliminaries}\label{sec:2}

In this section we outline the main steps in the proof of Theorem \ref{thm:xia} in \cite{Xia}. This not only serves as a guideline of general ideas to be appreciated in later sections, but also allows us to introduce some crucial results which will be repeatedly used later.

Throughout the paper, let $\C$ be an $[n,k,d]_q$ linear code. We always assume that its dual distance satisfies $\db \ge 5$. For any $a<b$, denote by $[a \isep b]$ the set of integers in the closed interval $[a,b]$. Let $\ep$ be the natural mapping $\ep: \Fq^{n} \to \bC^{n}$ obtained component-wise from the standard additive character on $\Fq$.

\subsection{Outline of the main steps in \cite{Xia}}

For a positive integer $p$, let $\Op$ be the set of maps $s: [1 \isep p] \to \D=\ep(\C)$ endowed with the uniform probability measure. Each $s \in \Op$ gives rise to a $p \times n$ matrix $\Phi(s)$ whose rows are listed as $s(1), \ldots, s(p)$. Let $\G(s)$ denote the Gram matrix of $\frac{1}{\sqrt{n}}\Phi(s)$, that is, $\G(s)=\frac{1}{n}\Phi(s)\Phi(s)^*$. For any positive integer $\l$, the $\l$-th moment of the spectral measure of $\G(s)$ is given by
\[A_{\l}(s)=\frac{1}{p} \tr \left(\G(s)^{\l}\right)=\frac{1}{pn^{\l}} \tr \left((\Phi(s) \Phi(s)^*)^{\l}\right). \]
Expanding the trace $\tr \left((\Phi(s) \Phi(s)^*)^{\l}\right)$, we have
$$A_{\l}(s)=\frac{1}{pn^{\l}} \sum_{\ga \in \Pi_{\l,p}} \om_\ga(s),$$
where $\Pi_{\l,p}$ is the set of all closed maps $\gamma$ from $[0 \isep \l]$ to $[1 \isep p]$ (``closed'' means $\gamma(0)=\gamma(\l)$), and
\begin{equation}\label{eq:Om(s)}
\om_\ga(s)=\prod_{j=0}^{\l-1} \langle s \circ \ga(j),s \circ \ga(j+1)\rangle.
\end{equation}
Here $ s \circ \gamma$ is the composition of the functions $s$ and $\gamma$, and $\langle \cdot,\cdot \rangle$ is the standard inner product. Taking expectation with respect to the probability space $\Op$ and rearranging the terms, \emph{the first main step} is to rewrite $\E(A_{\l}(s), \Op)$ as
\begin{eqnarray*} \label{2:exp} \E(A_{\l}(s), \Op)=\frac{1}{pn^{\l}} \sum_{\ga \in \Pi_{\l,p}/\Sp}\frac{p!}{(p-\vg)!}\E(\om_\ga(s), \Op),\end{eqnarray*}
where $\Pi_{\l,p}/\Sp$ is the set of equivalence classes of closed paths of $\Pi_{\l,p}$ under the equivalence relation
$$\ga_1 \sim \ga_2 \iff \ga_1=\sigma \circ \ga_2 \ \exists \, \sigma \in \Sp.$$
Here $\Sp$ is the permutation group on the set of integers $[1 \isep p]$.

It is easy to see that
\[\E(\om_\ga(s), \Op)=\E(\om_\ga(s), \Omega(V_{\ga})), \]
where
\[V_\gamma=\gamma\left([0 \isep l]\right), \quad v_\gamma= \# V_\gamma \leq \l,\]
and $\Omega(V_{\ga})$ is uniform probability space of all maps from $V_{\ga}$ to $\mathcal{D}$.

For simplicity, define
\begin{eqnarray} \Wg&=&\E(\om_\ga(s),\Omega(V_\ga)). \label{eq:wr}
\end{eqnarray}
\emph{The second main step} is to use properties of linear codes over finite fields to conclude that the quantity $\Wg$ is exactly the number of solutions $(t_0,t_1,\ldots,t_{\l-1}) \in [1\isep n]^\l$ satisfying the system of equations
$$\sum_{u \in I_a} (\bg_{t_u}-\bg_{t_{u-1}})=\bz, \ \forall 1 \leq a \leq v_\ga.$$
Here we write
$$V_\ga=\{z_a: 1 \leq a \leq v_\ga\}, \quad I_a=\ga^{-1}(z_a),\quad \forall \, a,$$
and $\bg_1,\bg_2,\ldots,\bg_n$ are the $n$ columns of a $k \times n$ generating matrix $G$ of the linear code $\C$.

Finally, in \emph{the last main step}, by some detailed analysis using number theory and graph theory, one can obtain (see \cite[Section IV]{Xia})
\begin{lemma} \label{le:0}
$$\Wg=\left\{\begin{array}{ll}
n^{\l-\vg+1} & \ga \in \Ga, \\
O\left(n^{\l-\vg}\right) & \ga \notin \Ga.
\end{array} \right.$$
Here $\Ga \subset \Pi_{\l,p}/\Sp$ is the subset of all closed paths that form double trees.
\end{lemma}
Armed with Lemma \ref{le:0}, we then can easily obtain the estimate
$$\E(A_\l(s),\Op)=\sum_{j=0}^{\l-1}\frac{y^j}{j+1}\binom{\l}{j}\binom{\l-1}{j}+O\left(\frac{\l^{\l+1}}{n}\right),$$
which is more than enough to prove Theorem \ref{thm:xia}.

\subsection{Two counting lemmas}

For $\gamma_1,\gamma_2 \in \Pi_{\l,p}$, we define
\begin{eqnarray} \Wgu:&=&\E(\om_{\ga_1}(s)\overline{\om_{\ga_2}(s)},\Op), \label{eq:wrr}
\end{eqnarray}
\begin{eqnarray*}
\Vgu:=V_{\ga_1} \cup V_{\ga_2}, & \vgu=\#\Vgu \, , \label{eq:vr+r} \\
\Vgi:=V_{\ga_1} \cap V_{\ga_2}, & \vgi=\#\Vgi \, \label{eq:vr-r}.
\end{eqnarray*}

We may reorder the indices as
$$\Vgi=\{z_a: a \in [1 \isep \vgi]\}, $$
$$V_{\ga_1} \setm V_{\ga_2}=\left\{z_a: a \in [\vgi+1 \isep v_{\ga_1}]\right\},$$
and
$$V_{\ga_2} \setm V_{\ga_1}=\left\{z_a: a \in [v_{\ga_1}+1 \isep \vgi]\right\}.$$
Let
$$I_a:=\gamma_1^{-1}(z_a), \quad J_a:=\gamma_2^{-1}(z_a) \quad \forall \, a.$$
Similar to the second main step in the previous subsection, expanding the expression $\om_{\ga_1}(s)\overline{\om_{\ga_2}(s)}$, collecting terms according to the sets $\Vgi, V_{\ga_1} \setm V_{\ga_2}$ and $V_{\ga_2} \setm V_{\ga_1}$ respectively and taking expectation over the probability space $\Op$, we can conclude that the term $\Wgu$ defined above is exactly the number of solutions $(t_0,\ldots,t_{\l-1},w_0,\ldots,w_{\l-1}) \in [1 \isep n]^{2\l}$ such that
\begin{align}
\sum_{u \in I_a} (\bg_{t_u}-\bg_{t_{u-1}})+\sum_{u \in J_a} (\bg_{w_{u-1}}-\bg_{w_u})&=\bz \quad \ \forall 1 \leq a \leq \vgi \, , \label{A}\\
\sum_{u \in I_b} (\bg_{t_u}-\bg_{t_{u-1}})&=\bz \quad \ \forall \vgi+1 \leq b \leq v_{\ga_1} \, , \label{B}\\
\sum_{u \in J_c} (\bg_{w_{u-1}}-\bg_{w_u})&=\bz \quad \ \forall v_{\ga_1}+1 \leq c \leq \vgu \, . \label{C}
\end{align}

We remark that in equations (\ref{A})--(\ref{C}), one equation is redundant, so we can remove any one equation without affecting the set of solutions. Using this we can obtain an estimate of $\Wgu$ as below:

\begin{lemma}\label{le:1}
If $\vgi \geq 1$, then
$$\Wgu=\begin{cases}
n^{2\l-\vgu+1} & \emph{ if } (\ga_1,\ga_2) \in \tilde{\Ga},\\
O(n^{2\l-\vgu}) & \emph{ if } (\ga_1,\ga_2) \notin \tilde{\Ga},
\end{cases}$$
where $\tilde{\Ga}$ is the set of all $(\ga_1,\ga_2) \in \Pi_{\l,p}^2$ such that the systems of equations (\ref{A})-(\ref{C}) for $\Wgu$ can be completely solved in the forms $t_u=t_{u-1}$ and $w_{v-1}=w_v$ for some $u$ and $v$.
\end{lemma}
\begin{proof}[Proof of Lemma \ref{le:1}]
Since $\vgi \geq 1$, it can be easily seen that the graph $\ga:=\ga_1 \cup \overline{\ga}_2$ is a closed path with $\vgi$ vertices and $2\l$ edges, where $\overline{\ga}_2$ is the closed path defined by reverting the directions of the edges of $\ga_2$ (after a cyclic relabelling of the vertices if necessary). The systems of equations (\ref{A})-(\ref{C}) for $\Wgu$ are precisely the same as those for $\Wg$. Therefore Lemma \ref{le:1} follows directly from Lemma \ref{le:0} on the estimate of $\Wg$.
\end{proof}

First notice that $W_\ga \ge 0$ for any $\ga$.  Armed with Lemmas \ref{le:0} and \ref{le:1}, we obtain

\begin{lemma}\label{le:2}
$$\Wgu-W_{\ga_1}W_{\ga_2}=\begin{cases}
0 & \emph{ if } \vgi \in \{0, 1\};\\
O(n^{2\l-\vgu}) & \emph{ if } \vgi \geq 2\,.
\end{cases}$$
\end{lemma}
\begin{proof}[Proof of Lemma \ref{le:2}]
We Write $\Wguu:=\Wgu-W_{\ga_1}W_{\ga_2}$. If $\vgi=0$, then equations in (\ref{A}) become empty, and equations in (\ref{B}) and (\ref{C}) are independent to each other, the number of solutions to which are $W_{\ga_1}$ and $W_{\ga_2}$ respectively. Hence $\Wgu=W_{\ga_1}W_{\ga_2}$ and so $\Wguu=0$.

If $\vgi=1$, then there is precisely one equation in (\ref{A}). We remove this equation without affecting $\Wgu$. The remaining equations are either in (\ref{B}) or in (\ref{C}), the number of solutions to which are exactly $W_{\ga_1}$ and $W_{\ga_2}$ respectively. Hence in this case we also have $\Wguu=0$.

Now assume $\vgi \geq 2$. If $(\ga_1,\ga_2) \in \tilde{\Ga}$, then each reduced equation is either of the form $t_u=t_{u-1}$ or $w_{v-1}=w_v$, which correspond to equations in either (\ref{B}) or (\ref{C}) respectively. Hence we still have $\Wguu=0$; otherwise if $(\ga_1,\ga_2) \notin \tilde{\Ga}$, then the result follows from the fact that $0 \leq \Wguu \leq \Wgu$ and Lemma \ref{le:1} on the estimate of $\Wgu$.
\end{proof}

\section{The $\l$-th Moment Estimate} \label{sec:3}

We use notation from Section \ref{sec:2}. Let $\C$ be an $[n,k,d]_q$ linear code with dual distance $\db \ge 5$. For a positive integer $p$, let $\Omega_{p, I}$ be the set of all injective maps $s: [1 \isep p] \to \D$ endowed with the uniform probability measure. Each $s \in \Omega_{p, I}$ gives rise to a $p \times n$ matrix $\Phi(s)$ whose rows are listed as $s(1), \ldots, s(p)$. Let $\G(s)$ denote the Gram matrix of $\frac{1}{\sqrt{n}}\Phi(s)$, that is, $\G(s)=\frac{1}{n}\Phi(s)\Phi(s)^*$.

Define
$$\G_I(s):=\sqrt{\frac{n}{p}}(\G(s)-I_p)=\sqrt{\frac{n}{p}}
\left(\frac{1}{n}\Phi(s)\Phi(s)^*-I_p\right),$$
and
$$\AIl(s):=\frac{1}{p}\text{Tr}(\G_I(s)^\l)=\frac{1}{p}
\left(\frac{n}{p}\right)^{\frac{\l}{2}}\text{Tr}\left(\left(
\frac{1}{n}\Phi(s)\Phi(s)^*-I_p\right)^\l\right).$$
We prove
\begin{theorem} \label{moment_SC}
If the conditions (i)-(iii) of Theorem \ref{thm:SC} are satisfied, then for $4 \leq \l^2 < \min\{p, \frac{N}{2}\}$, we have
\begin{equation*}
\E(\AIl(s), \OpI)=\begin{cases}
O_{\l}\left(\frac{c^\l}{\sqrt{p}}+\sqrt{\frac{p}{n}}\right) & \mbox{ if } \l\text{ is odd},\\
\frac{2}{\l+2}\binom{\l}{\l/2}+O_{\l}\left(\frac{c^\l}{p}+\frac{n}{N}+\frac{p}{n}
\right) & \mbox{ if } \l \text{ is even}.
\end{cases}
\end{equation*}
Here the constant implied in the big-O term depends only on the parameter $\l$.
\end{theorem}

Noting that the corresponding $\l$-th moments of the Wigner semicircle distribution are given by
$$A_{\l,\mathrm{SC}}=\begin{cases}
0 &\mbox{ if } \l\text{ is odd},\\
\frac{2}{\l+2}\binom{\l}{\l/2} &\mbox{ if } \l\text{ is even},
\end{cases}$$
hence by Theorem \ref{moment_SC}, for any fixed $\l$, as $n \to \infty$ and $p, \frac{N}{n}, \frac{n}{p} \to \infty$,  we have
$$\E(\AIl(s),\OpI) \to A_{\l,\mathrm{SC}}. $$

The rest of this section is devoted to a proof of Theorem \ref{moment_SC}.

\subsection{Problem Setting Up}



\begin{definition} \label{def:simple}
A closed path $\ga: [0\isep\l] \to [1 \isep p]$ is called \tb{simple} if it satisfies $\ga(j) \neq \ga(j+1) \ \all j$.
\end{definition}

Denote by $\Pi_{\l,p}'$ the set of all closed simple paths $\ga: [0\isep\l] \to [1 \isep p]$. This is a subset of $\Pi_{\l,p}$ appearing in Section \ref{sec:2}. Since all the diagonal entries of $\G_I(s)$ are zero, we can expand the expression of the trace in $\AIl(s)$ as
$$\AIl(s)=\frac{1}{p}\left(\frac{1}{np}\right)^{\frac{\l}{2}}\sum_{\ga \in \Pi_{\l,p}'} \om_\ga(s),$$
where $\om_\ga(s)$ is already defined in (\ref{eq:Om(s)}).

Similar to the first main step in Section \ref{sec:2} (see also Section III of \cite{Xia}) we can write
$$\E(\AIl(s), \OpI)=\frac{1}{p}\left(\frac{1}{np}\right)^{\frac{\l}{2}}\sum_{\ga \in \Pi_{\l,p}'/\Sp}\frac{p!}{(p-\vg)!}\E(\om_\ga(s), \OpI),$$
where
\[V_\gamma=\gamma\left([0 \isep l]\right), \quad v_\gamma=\# V_\gamma  \leq \l,\]
and $\Pi_{\l,p}'/\Sp$ is the set of equivalence classes of simple closed paths of $\Pi_{\l,p}'$ under the equivalence relation
$$\ga_1 \sim \ga_2 \iff \ga_1=\sigma \circ \ga_2 \ \exists \, \sigma \in \Sp.$$

We remark that
\[\E(\om_\ga(s), \OpI)=\E(\om_\ga(s), \Omega_{I}(V_{\ga})),\]
where $\Omega_{I}(V_{\ga})$ is the uniform probability space of all injective maps from $V_{\ga}$ to $\mathcal{D}$.

\subsection{Proof of Theorem \ref{moment_SC}}
Since $s$ is injective, $\ga$ is simple, so $\ga(j) \neq \ga(j+1) \ \all j$, from (\ref{1:srip}), we have
\begin{equation}\label{eq:1-StRIP2}
|\E(\om_\ga(s),\OpI)| \le c^\l n^{\frac{\l}{2}}.
\end{equation}
By Lemma \ref{4:le:3} in Section \ref{appendix} {\bf Appendix} we have another estimate:
\begin{equation}\label{eq:SCE2}
\E(\om_\ga(s),\OpI)=\E(\om_\ga(s),\Omega_p)+O_{\l} \left(\frac{n^{\l-\vg+2}}{N}\right).
\end{equation}
Define
\begin{equation*}
\beta_\ga:=\frac{1}{p}\left(\frac{1}{np}\right)^{\frac{\l}{2}}\frac{p!}{(p-\vg)!}\E(\om_\ga(s), \OpI),
\end{equation*}
hence we have
\begin{equation*} \label{eq:SCE16}
\E(\AIl(s), \OpI)=\sum_{\substack{\ga \in \Pi_{\l,p}'/\Sp}} \beta_\ga.
\end{equation*}
From (\ref{eq:1-StRIP2}), (\ref{eq:SCE2}) and Lemma \ref{le:0} we can summarize the estimates of $\beta_{\ga}$ as follows:
\begin{eqnarray*}
\begin{array}{llll}
(a).&\beta_\gamma \ll_{\l} \frac{c^{\l}}{\sqrt{p}} &:& v_\gamma <1+\frac{\l}{2},\\
(b).&\beta_\gamma \ll_{\l} \sqrt{\frac{p}{n}} \left(1+\frac{n}{N}\right) &:& v_\gamma >1+\frac{\l}{2},\\
(c).&\beta_\gamma \ll_{\l} \frac{1}{n} \left(1+\frac{n^2}{N}\right) &:& v_\gamma =1+\frac{\l}{2}, \ga \notin \Gamma,\\
(d).&\beta_\gamma =1+ O_{\l}\left(\frac{1}{p}+\frac{n}{N}\right) &:& v_\gamma =1+\frac{\l}{2}, \ga \in \Gamma.
\end{array}
\end{eqnarray*}
Note that (c) and (d) may appear only when $\l$ is even. Using
\begin{equation*}\label{bound}
\sum_{\substack{\ga \in \Pi_{\l,p}'/\Sp \\ \vg=v}} 1 < v^\l \leq \l^\l, \quad \forall \, v \le \l,
\end{equation*}
and the identity (see \cite{Xia} or \cite[Lemma 2.4]{RMT})
\begin{equation*} \label{eq:SCE35}
\sum_{\substack{\ga \in \Ga \\ \vg=1+\frac{\l}{2}}} 1=\frac{2}{\l+2}\binom{\l}{\frac{\l}{2}},
\end{equation*}
we obtain the desired estimates on $\E(\AIl(s), \OpI)$. This completes the proof of Theorem \ref{moment_SC}. $\qed$

\section{Proof of Theorem \ref{thm:SC}}\label{sec:4}
To complete the proof of Theorem \ref{thm:SC}, by the moment convergence theorem \cite[p.24]{RMT}, it suffices to prove the following result.
\begin{theorem}\label{VarSC2}
Assume the conditions of Theorem \ref{thm:SC} are satisfied. Then
$$\Var(\AIl(s),\OpI)=O_{\l}\left(\frac{c^{2\l}}{p^2}+\frac{1}{pn}+\frac{n}{pN}\right).$$
\end{theorem}
This section is devoted to a proof of theorem \ref{VarSC2}.

\subsection{Problem setting up}

By definition,
\begin{align*}
\Var(\AIl(s),\OpI)&=\E(|\AIl(s)|^2,\OpI)-|\E(\AIl(s),\OpI)|^2.
\end{align*}
Similar to the first main step in Section \ref{sec:2}, we can write


\begin{equation}\label{VarSC3}
\Var(\AIl(s),\OpI)=\sum_{(\ga_1,\ga_2) \in {\Pi'_{\l,p}}^2/\Sp} \frac{1}{p^2}\left(\frac{1}{np}\right)^\l\frac{p!}{(p-\vgu)!} \, \beta_{\ga_1,\ga_2},
\end{equation}
where
\begin{equation*}\label{b2}
\beta_{\ga_1,\ga_2}:=\E(\om_{\ga_1}(s)\overline{\om_{\ga_2}(s)},\OpI)-\E(\om_{\ga_1}(s),\OpI)\overline{\E(\om_{\ga_2}(s),\OpI)}.
\end{equation*}
Here ${\Pi'_{\l,p}}^2/\Sp$ denotes the set of equivalence classes of ordered pairs of simple closed paths in $\Pi'_{\l,p}$ under the equivalence relation
$$(\ga_{11},\ga_{21}) \sim (\ga_{12},\ga_{22}) \iff (\ga_{11},\ga_{21})=(\sigma \circ \ga_{12},\sigma \circ \ga_{22}) \ \exists \sigma \in \Sp.$$

For simplicity, for $\gamma_1,\gamma_2 \in \Pi_{\l,p}'$, we define
\begin{eqnarray*}
\Vgu:=V_{\ga_1} \cup V_{\ga_2}, & \vgu=\#\Vgu \, ,  \\
\Vgi:=V_{\ga_1} \cap V_{\ga_2}, & \vgi=\#\Vgi \, .
\end{eqnarray*}

\subsection{Study of $\beta_{\ga_1,\ga_2}$}

First, by the condition in (\ref{1:srip}), we easily obtain
\begin{equation}\label{b3}
|\beta_{\ga_1,\ga_2}| \le2 c^{2\l}n^{\l}.
\end{equation}
Next, we have the following estimation:
\begin{lemma}\label{le:5}
Assume $\db \geq 5$ and $4 \leq \l^2 \leq \frac{N}{8}$. Then
\begin{equation} \label{4:beta}
\beta_{\ga_1,\ga_2} \ll_{\l} n^{2 \l-\vgu+2} \left(\frac{1}{n^2}+\frac{1}{N}\right).
\end{equation}
\end{lemma}
\begin{proof}[Proof of Lemma \ref{le:5}]
If $\vgi \ge 1$, applying Lemma \ref{4:le:4} and Lemma \ref{4:le:3} in Section \ref{appendix} {\bf Appendix} directly to the terms $\E(\om_{\ga_1}(s)\overline{\om_{\ga_2}(s)},\OpI)$ and $\E(\om_{\ga_i}(s),\OpI)$ ($i=1,2$) respectively, then using Lemmas \ref{le:0}-\ref{le:2} in Section \ref{sec:2}, also observing that $v_{\ga_1}+v_{\ga_2}=\vgu+ \vgi \ge \vgu +1$, we obtain the desired result by a straightforward computation.

Now assume $\vgi = 0$. We remark that if we use the above approach, we can only obtain
\[ \beta_{\ga_1,\ga_2} \ll_{\l} n^{2 \l-\vgu+2} \left(\frac{n}{N}\right), \]
which falls short of our expectation (\ref{4:beta}). So we adopt a different method.

Denote
\[N_i=\#\Om_I(V_{\ga_i})=\frac{N!}{(N-v_{\ga_i})!}, \quad i=1,2,\]
and
\[N_0=\#\Om_I(\Vgu)=\frac{N!}{(N-v_{\ga_1,\ga_2})!}. \]

By using definition, we can rewrite $\beta_{\ga_1,\ga_2}$ as
\begin{eqnarray*}
\beta_{\ga_1,\ga_2} =A-B, \end{eqnarray*}
where
\begin{align}\label{b4}
A& =\left(1-\frac{N_0}{N_1N_2}\right) \E(\om_{\ga_1}(s)\overline{\om_{\ga_2}(s)},\OpI)\nono\\
B&=\frac{1}{N_1N_2}\left(\sum_{\substack{s \in \Om(\Vgu)\\s|_{V_{\ga_1}} \in \Om_I(V_{\ga_1})\\s|_{V_{\ga_2}} \in \Om_I(V_{\ga_2})}}- \sum_{s \in \Om_I(\Vgu)}\right)\om_{\ga_1}(s)\overline{\om_{\ga_2}(s)}. \nono
\end{align}

As for the first term $A$, since $0 \le \vgu=v_{\ga_1}+v_{\ga_2} \le 2 \l$, we have  $1-\frac{N_0}{N_1N_2} \ll_{\l} \frac{1}{N}$. By Lemma \ref{4:le:4} and noting that
\[\Wgu = W_{\ga_1} W_{\ga_2} \le n^{\l-v_{\ga_1}+1}n^{\l-v_{\ga_2}+1}=n^{l-v_{\ga_1,\ga_2}+2},\]
we can obtain easily
\[A \ll_{\l} \frac{1}{N} \left(n^{2 \l -\vgu+2}+ \frac{n^{2 \l -\vgu+2}}{N}\right) \ll_{\l} n^{2 \l -\vgu+2} \left(\frac{1}{N}\right). \]

As for $B$, first, we can rewrite it as
\[B=\frac{1}{N_1N_2} \sum_{s \in \Om_I(V_{\ga_1}) \times \Om_I(V_{\ga_2})\setm\Om_I(\Vgu)}\om_{\ga_1}(s)\overline{\om_{\ga_2}(s)}. \]
Here subscript means that we sums over all $s \in \Om_I(V_{\ga_1}) \times \Om_I(V_{\ga_2})$ such that there are $a \in V_{\ga_1}$ and $b \in V_{\ga_2}$ with $s(a)=s(b)$.

Let $Q=\{(a,b): a \in V_{\ga_1}, b \in V_{\ga_2}\}$. For any non-empty subset $U \subset Q$, we can define corresponding new maps $\ga_{1U}$ and $\ga_{2U}$ by gluing the vertices corresponding to $a_k$ and $b_k$ together whenever $(a_k,b_k) \in U$. For these new maps, clearly we have
$$v_{\ga_{1U},\ga_{2U}} \leq \vgu-1.$$
Moreover, since $\ga_{1U}$ and $\ga_{2U}$ share the new vertex formed by gluing $a_k$ and $b_k$ together, we also have $v_{\ga_{1U}\cap\ga_{2U}} \geq 1$. Hence we can apply Lemma \ref{4:le:4} and Lemma \ref{le:1} to obtain
\begin{eqnarray*} \left|\sum_{s \in \Om_I(V_{\ga_{1U},\ga_{2U}})} \om_{\ga_{1U}}(s)\overline{\om_{\ga_{2U}}(s)}\right| &\ll_{\l} & N^{v_{\ga_{1U},\ga_{2U}}} \left|\E(\om_{\ga_{1U}}(s)\overline{\om_{\ga_{2U}}(s)},\OpI) \right| \\
& \ll_{\l}& N^{\vgu} n^{2\l-\vgu+2}\left(\frac{1}{N}+\frac{n}{N^2}\right).
\end{eqnarray*}
Then by the inclusion-exclusion principle, we conclude that
\begin{align}\label{SE8}
\left|\sum_{s \in \Om_I(V_{\ga_1}) \times \Om_I(V_{\ga_2})\setm\Om_I(\Vgu)}\om_{\ga_1}(s)\overline{\om_{\ga_2}(s)}\right| & \leq \sum_U \left|\sum_{s \in \Om_I(V_{\ga_{1U},\ga_{2U}})} \om_{\ga_{1U}}(s)\overline{\om_{\ga_{2U}}(s)}\right|\nono\\
& \ll_{\l} N^{\vgu} n^{2\l-\vgu+2}\left(\frac{1}{N}+\frac{n}{N^2}\right). \nono
\end{align}
From this we obtain
\[B \ll_{\l} n^{2\l-\vgu+2}\left(\frac{1}{N}+\frac{n}{N^2}\right). \]
Combining the estimates of $A$ and $B$ yields the desired result for $\beta_{\ga_1,\ga_2}$. This completes the proof of Lemma \ref{le:5}.
\end{proof}

\subsection{Proof of Theorem \ref{VarSC2}}

For simplicity, define
\[\alpha_{\ga_1,\ga_2}=\frac{1}{p^2}\left(\frac{1}{np}\right)^\l\frac{p!}{(p-\vgu)!} \, \beta_{\ga_1,\ga_2}. \]
From (\ref{b3}) and Lemma \ref{le:5} we summarize the estimates of $\alpha_{\ga_1,\ga_2}$ as follows:
\begin{align} \label{4:alpha1}
\alpha_{\ga_1,\ga_2} &\ll_{\l} c^{2 \l} p^{\vgu-\l-2},\\
\label{4:alpha2}
\alpha_{\ga_1,\ga_2} &\ll_{\l} \left(pn^{-1}\right)^{\vgu-\l-2} \left(n^{-2}+N^{-1}\right).
\end{align}
We split $\Var(\AIl(s),\OpI)$ in (\ref{VarSC3}) into two terms
\begin{equation}\label{VarSC4}
\Var(\AIl(s),\OpI)=\sum_{\substack{(\ga_1,\ga_2) \in {\Pi'_{\l,p}}^2/\Sp\\\vgu \le \l}} \alpha_{\ga_1,\ga_2}+\sum_{\substack{(\ga_1,\ga_2) \in {\Pi'_{\l,p}}^2/\Sp\\\vgu \geq \l+1}} \alpha_{\ga_1,\ga_2}.
\end{equation}
For the first term, using (\ref{4:alpha1}) and the trivial bound
\begin{equation*}\label{bound2}
\sum_{\substack{(\ga_1,\ga_2) \in {\Pi'_{\l,p}}^2/\Sp \\ \vgu=v}} 1 < v^{2\l} \leq (2\l)^{2\l},
\end{equation*}
we easily obtain
\begin{equation}\label{VarSC5}
\sum_{\substack{(\ga_1,\ga_2) \in {\Pi'_{\l,p}}^2/\Sp\\\vgu \le \l}} \alpha_{\ga_1,\ga_2} \ll_{\l} \frac{c^{2\l}}{p^2}\,\, .
\end{equation}
For the second term of (\ref{VarSC4}), using (\ref{4:alpha2}) we can also obtain
\begin{align}\label{VarSC6}
\sum_{\substack{(\ga_1,\ga_2) \in {\Pi'_{\l,p}}^2/\Sp\\\vgu \geq \l+1}} \alpha_{\ga_1,\ga_2} &\ll_{\l} \frac{1}{p}\left(\frac{1}{n}+\frac{n}{N}\right).
\end{align}

Putting (\ref{VarSC5}) and (\ref{VarSC6}) into (\ref{VarSC4}) gives the desired result for $\Var(\AIl(s),\Om_{p,I})$. This completes the proof of Theorem \ref{VarSC2}. Now Theorem \ref{thm:SC} is proved. $\qed$

\section{Appendix: two lemmas}\label{appendix}

\subsection{Some lemmas}

Now we prove two technical lemmas which were used in Sections \ref{sec:3} and \ref{sec:4} before.

\begin{lemma}\label{4:le:3}
Assume that $\db \geq 5$. Then for all $\l$ such that $4 \leq \l^2 \leq \frac{N}{2}$, we have
\begin{equation*}
\E(\om_\ga(s),\OpI)=\E(\om_\ga(s),\Omega_p)+O_{\l}\left(\frac{n^{\l-\vg+2}}{N}\right).
\end{equation*}
Here the constant implied in the symbol $O_{\l}$ depends only on the parameter $\l$.
\end{lemma}
\begin{proof}
First, note that
\[|\E(\om_\ga(s),\OpI)|=|\E(\om_\ga(s),\OIg)|,\]
$\OIg$ is the set of all injective maps $s: \Vg \to \D$ endowed with the uniform probability. Define
\begin{equation*}
\tE(\om_\ga(s), \OIg):=\frac{\sum_{s \in \OIg} \om_\ga(s)}{N^{\vg}}.
\end{equation*}
Noting that
\begin{eqnarray} \E(\om_\ga(s), \OIg) &=&\tE(\om_\ga(s), \OIg)\frac{N^{\vg}}{N(N-1)(N-2)\cdots(N-\vg+1)} \nonumber\\
& =& \tE(\om_\ga(s), \OIg) \left(1+O\left(\frac{\l^2}{N}\right)\right), \label{eq:E1E2}
\end{eqnarray}
to prove Lemma \ref{4:le:3}, it suffices to study $\tE(\om_\ga(s), \OIg)$. We write

\begin{equation} \label{eq:SCE3}
\tE(\om_\ga(s), \OIg)=\frac{\sum_{s \in \Og} \om_\ga(s)-\sum_{s \in \Og \setm \OIg} \om_\ga(s)}{N^{\vg}}.
\end{equation}
Here $\Og$ is the set of all maps $s: \Vg \to \D$ endowed with the uniform probability. The first term is precisely $\Wg$ defined in (\ref{eq:wr}). As for the second term, the condition $s \in \Og \setminus \OIg$ is equivalent to $s$ being not injective, that is, there exist $a \ne b \in \Vg$ such that $s(a)=s(b)$. Denote by $\Om_{(a,b)}$ the set of all $s \in \Og$ such that $s(a)=s(b)$. We may order the set $V_\gamma$ as $V_\gamma=\left\{z_i: 1 \le i \le v_\gamma\right\}$ and define $P=\{(z_i,z_j): 1 \leq i < j \leq \vg\}$. Using
$$\Om(\Vg)\setm\OIg=\cup_{(a,b) \in P} \, \Om_{(a,b)},$$
and the inclusion-exclusion principle, we have
\begin{align*}\label{noinject}
\left|\sum_{s \in \Om(\Vg)\setm\OIg} \om_\ga(s)\right|
&\leq \sum_{t=1}^{|P|} \sum_{\substack{(a_1,b_1),\cdots,(a_t,b_t) \in P\\ \text{distinct}}}\left|\sum_{s \in \cap_{m=1}^t \Om_{(a_m,b_m)}} \om_\ga(s)\right|.
\end{align*}

A little thought reveals that the inner summand $\sum_{s \in \cap_{m=1}^t \Om_{(a_m,b_m)}} \om_\ga(s)$ corresponds to the quantity $W_{\gamma_T}$ defined in (\ref{eq:wr}), where the graph $\gamma_T$ is obtained from $\gamma$ by gluing the vertices $a$ and $b$ together for all pairs $(a,b)$ inside the set $T=\left\{(a_m,b_m): 1 \le m \le t \right\}$. More precisely, let $v_{\gamma_T}$ be the number of vertices of $\gamma_T$, then
\begin{equation*}\label{intersect}
\frac{1}{N^{v_{\gT}}} \left|\sum_{s \in \cap_{m=1}^t \Om_{(a_m,b_m)}} \om_\ga(s)\right|= W_{\gT}.
\end{equation*}
Obviously $v_{\gamma_T} \le v_{\gamma}-1$. Applying Lemma \ref{le:0} on $W_{\gamma_T}$ directly, we obtain
\begin{align*}
\left|\sum_{s \in \Om(\Vg)\setm\OIg} \om_\ga(s)\right|
& \ll 2^{\l^2} n^{\l+1} \left(\frac{N}{n}\right)^{v_\gamma-1}.
\end{align*}
Inserting this into (\ref{eq:SCE3}), we obtain
$$\tE(\om_\ga(s), \OIg)=\Wg+O_{\l}\left(\frac{n^{\l-\vg+2}}{N}\right). $$
Noting the relation (\ref{eq:E1E2}), we obtain the desired estimate on $\E(\om_\ga(s), \OIg)$. This completes the proof of Lemma \ref{4:le:3}.
\end{proof}

\begin{lemma} \label{4:le:4} Assume $\db \geq 5$ and $4 \leq \l^2 \leq \frac{N}{8}$. Then
$$\E(\om_{\ga_1}(s)\overline{\om_{\ga_2}(s)},\OpI)=\Wgu+
O_{\l}\left(\frac{n^{2\l-\vgu+2}}{N}\right),$$
where $\Wgu$ is defined in (\ref{eq:wrr}).
\end{lemma}
The proof of Lemma \ref{4:le:4} is very similar to that of Lemma \ref{4:le:3}, by using the inclusion-exclusion principle to translate from the set $\OpI$ to $\Op$. For the sake of simplicity, we omit the details.

\section{Conclusion}\label{sec:5}
In this paper, we investigate conditions under which linear codes possess the group randomness property with respect to Wigner's semicircle law. This is analogous to previous work on the group randomness of linear codes with respect to the Marchenko-Pastur law. Several interesting questions arise during the course of writing this paper, and we hope to stress these questions in the future.

\begin{enumerate}

\item While we have proved the convergence in probability in Theorem \ref{thm:SC}, our numerical experiments seem to indicate that the convergence is quite fast with respect to $n$, the length of the codes. Can one prove something substantial, say a rate of convergence in probability in the order of $O\left(n^{-\epsilon}\right)$ for some $\epsilon >0$? This question also remains interesting for the group randomness of linear codes with respect to the Marchenko-Pastur law.

\item How about other group randomness properties for linear codes, and how these properties may reflect the algebraic properties of the underlying codes? There has been some very interesting recent work on pseudo-Wigner matrices from linear codes \cite{Sol, Sol2}, and these may lead the door open for further investigations.




\end{enumerate}

\subsection*{Acknowledgment}
The research of M. Xiong was supported by RGC grant number 16303615 from Hong Kong.

\end{document}